\documentclass[aps,prl,twocolumn,showpacs,superscriptaddress,10pt]{revtex4-1}
\usepackage[utf8]{inputenc} % Umlaute : ä statt \"a usw., ß statt \ss{}
\usepackage[T1]{fontenc} % Silbentrennung mit Umlauten
\usepackage{textcomp,lmodern} %statt "ae" und "aecompl", die die Suche nach Umlauten in PDFs unmöglich machen;
\usepackage{amssymb,amsmath,amsfonts,amsthm,bbm} %für das Winkelzeichen u.a.

\begin{document}
\newcommand{\kt}[1]{\ensuremath{|#1\rangle}}
\newcommand{\br}[1]{\ensuremath {\langle #1|}}
\newcommand{\HS}{\mathcal{H}}
\newcommand{\bk}[2]{\ensuremath {\langle #1|#2 \rangle}}
\newcommand{\ckt}[1]{\ensuremath{|#1\}}}
\newcommand{\cbr}[1]{\ensuremath {\{#1|}}
\newcommand{\R}{\mathbb{R}}
\renewcommand{\v}[1]{\mathbf{#1}}
\newcommand{\m}[1]{\mathrm{#1}}
\newcommand{\vp}{\mathbf{p}}
\newcommand{\AO}{\mathcal{A}}
\newcommand{\hS}{\hat{S}}  
\newcommand{\mS}{\mathrm{S}}

%KSR-Notation
\newcommand{\anfEngl}[1]{``#1''} %englische Anführungszeichen
\newcommand{\ens}[0]{\ensuremath} %Kurzschreibweise für mathematische Definitionen
\newcommand{\ket}[1]{\ens{|#1\rangle}} %Ket-Vektor
\newcommand{\bra}[1]{\ens{\langle#1|}} %Bra-Vektor
\newcommand{\x}[0]{\ens{\otimes}} %Tensorprodukt
\newcommand{\isom}[0]{\ens{\cong}} %isomorph
\newcommand{\Mge}[2]{\ens{\left\lbrace #1|\,#2 \right\rbrace}} %Mengendefinition (ggf. ":" statt "|")
\newcommand{\Mg}[1]{\ens{\lbrace #1 \rbrace}} %Def. nur wg. der Klammern
\newcommand{\MgE}[1]{\ens{\Mg{1,\dots,#1}}} % Mengen der Form {1,...,#1}
\newcommand{\betrag}[1]{\ens{|#1|}} %Betrag
\newcommand{\Fkt}[3]{\ens{#1 : #2 \rightarrow #3}} %Funktion : z. B. "f : X -> Y"
\newcommand{\Eins}[0]{\ens{\mathbbm{1}}}
\newcommand{\iE}[0]{\ens{\mathrm{i}}} %imaginäre Einheit
\renewcommand{\phi}[0]{\ens{\varphi}} %schöneres Phi
\newtheorem{Definition}{Definition} %newtheorem.pdf, Abschnitt 2.2 (S. 4)
\newtheorem{Lemma}[Definition]{Lemma} %Lemma, Hilfssatz
\newtheorem{Theorem}[Definition]{Theorem}
\newtheorem{Corollary}[Definition]{Corollary} %Korollar, Folgerung
\newcommand{\cA}[0]{\ens{\mathcal{A}}}
\newcommand{\cB}[0]{\ens{\mathcal{B}}}
\newcommand{\cH}[0]{\ens{\mathcal{H}}}
\newcommand{\cM}[0]{\ens{\mathcal{M}}}
\newcommand{\N}[0]{\ens{\mathbb{N}}}
\newcommand{\C}[0]{\ens{\mathbb{C}}}
\newcommand{\diag}[0]{\mathrm{diag}}

\title{Observables Can Be Tailored to Make Any Pure State Entangled (or Not)}

%\author{N.L.~Harshman}
%\affiliation{Department of Physics, American University\\
%4400 Massachusetts Ave., NW, Washington, DC 20016-8058}
\author{N.\,L.\,Harshman}
\affiliation{Department of Physics, American University, 4400 Massachusetts Ave., NW, Washington, DC 20016-8058}

\author{Kedar S. Ranade}
\affiliation{Institut f\"ur Quantenphysik, Universit\"at Ulm, Albert-Einstein-Allee 11,
  D-89069 Ulm, Deutschland/Germany}

\begin{abstract}
We show that for a finite-dimensional Hilbert space, there exist observables that induce a tensor product structure such that the entanglement properties of any pure state can be tailored.  In particular, we provide an explicit, finite method for constructing observables in an unstructured $d$-dimensional system so that an arbitrary known pure state has any Schmidt decomposition with respect to an induced bipartite tensor product structure.  In effect, this article demonstrates that in a finite-dimensional Hilbert space, entanglement properties can always be shifted from the state to the observables and all pure states are equivalent as entanglement resources in the ideal case of complete control of observables.
\end{abstract}

\maketitle

The entanglement of a quantum state is only defined with respect to a tensor product structure within the Hilbert space that represents the quantum system. In turn, a tensor product structure of the Hilbert space is induced by the algebra of observables.  Zanardi and colleagues~\cite{zanardi_virtual_2001,zanardi_quantum_2004} have provided criteria for the algebra of observables of a finite-dimensional system to induce a tensor product structure.  The algebra of observables must be partitioned into subalgebras that satisfy two mathematical requirements, the subalgebras must be independent and complete (see Corollary \ref{ZanardiMulti} for a precise formulation of Zanardi's Theorem), and one physical requirement, the subalgebras must be locally accessible.  Such observable-induced partitions of the Hilbert space have been referred to as virtual subsystems and can be thought of as a generalization from entanglement between subsystems to entanglement between degrees of freedom (see also~\cite{barnum,delatorre}).  This mathematical framework has found applications to studies of multi-level encoding~\cite{mle}, decoherence~\cite{deco}, operator quantum error correction~\cite{oqec}, entanglement in fermionic systems~\cite{fer}, single-particle entanglement~\cite{cunha, osid}, and entanglement in scattering systems~\cite{scat}.

In this Letter, we extend this mathematical framework and prove what we call the \emph{Tailored Observables Theorem} (Theorem \ref{tailor}): observables can be constructed such that \emph{any} pure state in a finite-dimensional Hilbert space $\HS=\mathbb{C}^d$ has \emph{any} amount of entanglement possible for \emph{any} given factorization of the dimension $d$ of $\HS$. This means all pure states are equivalent as entanglement resources in the ideal case of complete control of observables.  To establish the framework, we provide a brief, relatively self-contained introduction to Zanardi's Theorem 
and obtain some necessary preliminary results about observable algebras in finite dimensions.  We then prove Theorem \ref{tailor}, which applies to bipartite tensor product structures, and present an illustrative example.  We will also provide a corollary of the theorem (Corollary \ref{tailor2}) applied to multipartite tensor product structures.  Before delving into the technical details, we present a more intuitive discussion of this result.

Consider a finite-dimensional Hilbert space $\HS=\mathbb{C}^d$.  The full matrix algebra $\mathcal{M}_d$ of complex $d \times d$ matrices acting on $\mathbb{C}^d$ will contain the representations of measurements and interactions that act on the states of the physical system, i.e.\ the algebra of observables $\mathcal{A}\subseteq \mathcal{M}_d$.   If the dimension $d$ can be factorized as $d= k_1 \cdot k_2 \cdot \ldots \cdot k_N$, this Hilbert space could represent states of a quantum system composed from $N$ subsystems each represented by Hilbert spaces $\HS_i = \mathbb{C}^{k_i}$.  For example, if $d=8$ then the system could be constructed from one qubit and one ququart ($N=2$, $k_1=2, k_2=4$) or three qubits ($N=3$, $k_1=k_2=k_3=2$).  By the process of subsystem composition, the total Hilbert space $\HS$ would inherit a tensor product structure $\HS \isom \bigotimes_{i=1}^N \HS_i = \bigotimes_{i=1}^N \mathbb{C}^{k_i}$.    If additionally the $N$ subsystems are localized into $N$ space-time separated regions, each subsystem would have an observable subalgebra $\mathcal{A}_i\subseteq\mathcal{M}_{k_i}$ that is operationally independent, and there would be exact correspondence between locality in space-time and locality with respect to the tensor product structure.

In contrast, the same Hilbert space $\HS=\mathbb{C}^d$ could represent a quantum system with no \emph{a priori} quantum subsystems such as the lowest $d$ energy levels of an harmonic oscillator.  However, even in such a system with no `natural' subalgebras of observables with which to partition the Hilbert space, the total observable algebra $\mathcal{A}=\mathcal{M}_d$ can `artificially' be divided into subalgebras $\mathcal{A}_i=\mathcal{M}_{k_i}$ that satisfy Zanardi's Theorem.  We provide an explicit constructive method for generating these subalgebras from a finite set of operators.  The generators may look somewhat arbitrary in the unstructured Hilbert space, but they have the correct properties to rigorously define locality, separability and entanglement.  By tailoring these subalgebras to a particular pure state, any entanglement properties for that state can be achieved, including maximal entanglement for any pure state, where the maximum depends on the the dimension.

In some sense, our results are an immediate consequence of the fact that all Hilbert spaces with the same dimension are isomorphic.  If we have a pure state of an unstructured $d$-level system $\kt{\phi}\in\HS=\mathbb{C}^d$ and a pure state $\kt{\phi'}$ of a $d$-level system with a tensor product structure $\HS'=\bigotimes_{i} \HS'_i$, then there will always exist a unitary map $U: \HS\rightarrow \HS'$ such that $U\kt{\phi}=\kt{\phi'}$.  Additionally, if there are local observable algebras $\AO'_i$ acting on each $\HS_i^\prime$, they can be mapped back to algebras $\AO_i= U^\dag\AO'_i U $ that act on $\HS$.  This article explains the conditions on the algebras for this map to exist and to induce a tailored tensor product structure.  Further, we show that each subalgebra $\mathcal{A}_i=\mathcal{M}_{k_i}$ corresponding to the  $\HS_i=\mathbb{C}^{k_i}$ factor of the tensor product can be finitely generated by the $k_i$-dimensional matrix representations of a basis for the $\mathfrak{su}(2)$ Lie algebra.  For a two-dimensional representation, the subalgebra is the Pauli operators, whose completeness as an algebraic basis for $\mathcal{M}_2$ is well-known, but we include a proof that this property holds true for all finite-dimensional representations of $\mathfrak{su}(2)$.  An alternate approach to generalizing Pauli operators to higher dimensions is taken in \cite{mle}.

Entanglement is detected as coherences between non-local observables, for example in the form of Bell-type inequalities.  This gives another way to look at our results: we give a method to construct observables that induce a notion of locality such that the intrinsic self-coherence of any pure state can be exploited as entanglement.  That the choice of observables, or equivalently degrees of freedom, used to describe a system can be tailored to suit a particular need is of course well-known in classical and quantum physics.  Certain observables can be preferred because of the form of interactions, the sources of error and decoherence, the physical accessibility of measurement and control, or for other reasons.  For example, in bound states of a proton and an electron, energy eigenstates are unentangled with respect to the tensor product induced by the center-of-mass/relative observables, but they are highly entangled with respect to the particle observables~\cite{tommasini_hydrogen_1998}.  By turning on external fields, one can induce entanglement between the center-of-mass and relative observables, i.e.~correlate electronic states to motional states of the atom. In this case, and in many others, the presence and dynamics of entanglement can serve as a proxy for the effect of interactions.

We now proceed with the formal statement of our results, beginning with some mathematical terminology for operator algebras.
We consider an (associative) algebra $\cM$ and a subalgebra $\cA \subseteq \cM$ thereof. The
\emph{centralizer} (or, in operator theory, the \emph{commutant}) $\cA^\prime$ of $\cA$ in $\cM$ is
defined as the set of operators in $\cM$ which commute with every element in $\cA$, i.\,e.
\begin{equation}
  \cA^\prime := \Mge{B \in \cM}{(\forall A \in \cA)(AB = BA)};
\end{equation}
this again is an algebra. Of principal interest here are the full matrix algebras
$\cM_d$ on finite-dimensional Hilbert spaces $\C^d$, which can be identified with the $d \times d$ matrices
$\C^{d \times d} = \Mge{(a_{ij})_{i,j = 1}^{d}}{(\forall i,\,j \in \MgE{d})(a_{ij} \in \C)}$.
The algebras $\cM_d$ are \emph{simple}, i.\,e., they contain no non-trivial two-sided ideals,
and \emph{central}, i.\,e. $\cM_d^\prime = \C \Eins_{\cM_d}$. We shall make use of the following lemma
\cite[p. 115]{Knapp}.
\begin{Lemma}[Double centralizer theorem]\label{DCT}\hfill\\
  Consider a finite-dimensional central simple algebra $\cM$ over an arbitrary field.
  Let $\cA \subseteq \cM$ be a simple subalgebra of $\cM$. Then, the centralizer $\cA^\prime$ is
  simple, $\cA^{\prime\prime} = \cA$ and $\dim \cA \cdot \dim \cA^\prime = \dim \cM$.
\end{Lemma}
Given two subalgebras $\cA,\,\cB \subseteq \cM$, we can construct a new algebra in two
ways: (i) take the tensor product $\cA \x \cB$ and (ii) take the subalgebra of $\cM$ generated by
$\cA$ and $\cB$, i.\,e., the smallest subalgebra of $\cM$ containing both $\cA$ and $\cB$, and
which we shall denote by $\cA \vee \cB$. These two constructions somewhat resemble the notion of internal and external
direct sums: in the first case, we consider $\cA$ and $\cB$ to be completely unrelated to each other,
so that this tensor product may be called \anfEngl{external}, while in the second (\anfEngl{internal})
case, we view them as substructures of the larger structure $\cM$, but where we have to introduce
some \anfEngl{non-overlapping} condition. This insight may be useful in understanding the following
theorem, the bipartite case of Zanardi's Theorem~\cite{zanardi_quantum_2004,zanardi_virtual_2001}. \renewcommand{\labelenumi}{(\roman{enumi})}
\begin{Theorem}[Induced tensor product structures]\label{Zanardi}\hfill\\
  Consider the full matrix algebra $\cM_d$ on the finite-dimensional Hilbert space $\cH = \C^d$ and
  two subalgebras $\cA$ and $\cB$ of $\cM_d$, for which there hold 
  \begin{enumerate}
    \item Independence: $[\cA,\,\cB] = \Mg{0}$, i.\,e. $[a,\,b] = 0$ for all $a \in \cA$ and $b \in \cB$; and
    \item Completeness: $\cA \x \cB \isom \cA \vee \cB = \cM_d$.
  \end{enumerate}
  Then, $\cA$ and $\cB$ induce a tensor product structure on $\C^d$, i.e., there exist two Hilbert spaces
  $\C^k$ and $\C^l$, $d = k \cdot l$, and a unitary mapping $\Fkt{U}{\C^k \x \C^l}{\C^d}$, such that
  $\cA = U(\cM_k \x \Eins_l)U^\dagger$ and $\cB = U(\Eins_k \x \cM_l)U^\dagger$.
  In particular, $\cA$ and $\cB$ are isomorphic to $\cM_k$ and $\cM_l$, respectively.
\end{Theorem}
\begin{proof}
  The algebra $\cA$ is unitarily equivalent to the direct sum of irreducible parts $\cM_{k_i}$
  with $k_1,\,k_2,\,\dots,\,k_n \in \N$, counted with multiplicities~$l_i$ \cite[Th. I.11.9, pp. 53--54]{Takesaki}; note that $\sum_{i = 1}^{n} k_i \cdot l_i = d$.
  In~other words, there exists a unitary operator $\Fkt{U}{\bigoplus_{i = 1}^{n} (\C^{k_i} \x \C^{l_i})}{\C^d}$
  such that $\cA = U\bigl[\bigoplus_{i = 1}^{n} \cM_{k_i} \x \Eins_{l_i}\bigr]U^\dagger$.
  By condition (i), we then have $\cB \subseteq \cA^\prime = U\bigl[\bigoplus_{i = 1}^{n} \Eins_{k_i} \x \cM_{l_i}\bigr]U^\dagger$ and $\cA \vee \cB \subseteq \cA \vee \cA^\prime \linebreak = U[\bigoplus_{i = 1}^{n} \cM_{k_i} \x \cM_{l_i}]U^\dagger$, $\dim(\cA \vee \cA^\prime) = \sum_{i = 1}^{n} (k_i \cdot l_i)^2$. In view of
  $\dim \cM_d = d^2$, condition (ii) implies $n = 1$, so that $\cA = U(\cM_k \x \Eins_l)U^\dagger$ with
  $k \cdot l = d$, where we removed the subscripts. As $\cA$ is simple, \mbox{$\dim \cA \cdot \dim \cA^\prime = \dim \cM_d$} by lemma~\ref{DCT}. If the inclusion $\cB \subseteq \cA^\prime$ were proper, we had
  $\dim \cB < \dim \cA^\prime$ and $\dim \cA \cdot \dim \cB < \dim \cM_d$ in contradiction to assumption~(ii).
\end{proof}
The formulation in \cite{zanardi_quantum_2004} also includes the physics requirement of local accessibility; this is a matter important for practical feasibility but does not affect the mathematical structure.  Theorem \ref{Zanardi} can be extended to the mulitpartite case.  
\begin{Corollary}[Zanardi's Theorem]\label{ZanardiMulti}\hfill\\
  Consider algebras $\cA_1,\,\dots,\,\cA_N \subseteq \cM_d$, such that
  \begin{enumerate}
    \item Independence: $[\cA_i,\,\cA_j] = \Mg{0}$ for all pairs $i \neq j$
    \item Completeness:    $\bigotimes_{i = 1}^{N} \cA_i \isom \bigvee_{i = 1}^{N} \cA_i = \cM_d$.
    \end{enumerate}
  Then there exist Hilbert spaces $\C^{k_1},\,\dots,\,\C^{k_N}$ with \mbox{$d = \prod_{i = 1}^{N} k_i$},
  and a unitary mapping $\Fkt{U}{\bigotimes_{i = 1}^{N} \C^{k_i}}{\C^d}$, such that
  $\cA_i = U(\bigotimes_{j = 1}^{k-1} \Eins_{k_j} \x \cM_{k_i} \x \bigotimes_{j = k+1}^{N} \Eins_{k_j})U^\dagger$
  $\isom \cM_{k_j}$ for all $i \in \MgE{N}$.
\end{Corollary}
\begin{proof}
  Set $\cA := \cA_1$ and $\cB := \bigvee_{i = 2}^{N} \cA_i$ and proceed by induction using theorem \ref{Zanardi}.
\end{proof}

\par Our subsequent construction relies on a property of representations of $\mathfrak{su}(2)$ on
$\C^{2s+1}$, which we will establish briefly.
The Lie algebra $\mathfrak{su}(2)$ can be defined as the complex linear hull of three (abstract) generators
$\hat{S}_x$, $\hat{S}_y$ and $\hat{S}_z$ which fulfill the commutation relation $[\hat{S}_i, \hat{S}_j] =
\iE \hbar \varepsilon_{ijk} \hat{S}_k$, where we set $\hbar = 1$ for the rest of this letter. It is well-known
that the eigenvalues of the representing operators $\Mg{\mS_x,\,\mS_y,\,\mS_z}$ on $\C^{2s+1}$
(the \mbox{spin-$s$} representation),
$s \in \Mg{\frac{1}{2},\,1,\,\frac{3}{2},\,2,\,\dots}$, have eigenvalues $-s,\,-s+1,\,\dots,\,+s$.
In the following, we will work with the equivalent set of operators $\Mg{\hat{S}_z,\,\hat{S}_+,\,\hat{S}_-}$, where
$\hat{S}_\pm := \hat{S}_x \pm \iE \hat{S}_y$. We construct representation matrices
$\mS_i = (\mS^{(i)}_{m,m^\prime})_{m,m^\prime = -s}^{s} \in \cM_{2s+1}$ for $i \in \Mg{z,+,-}$~by
\begin{subequations}\begin{align}
  \mS^{(z)}_{m,m^\prime}   &= m \delta_{m,m'} \quad\mbox{and}\\
  \mS^{(\pm)}_{m,m^\prime} &= \sqrt{s(s+1)-m(m\pm1)} \delta_{m,m'\pm 1},
\end{align}\end{subequations}
where we choose $\mS_z$ to be diagonal; we can always achieve this by some unitary transformation.  For $s=\frac{1}{2}$, the basis generators are proportional to the Pauli matrices $\mS_i = \frac{1}{2}\sigma_i$.
As a Lie algebra with the Lie bracket composition rule, the operators $\hat{S}_x$, $\hat{S}_y$
and $\hat{S}_z$ generate a (proper) subalgebra of $\cM_{2s+1}$ isomorphic to $\mathfrak{su}(2)$. We will now show
that, seen as a matrix algebra with matrix multiplication, they generate the full algebra $\mathcal{M}_{2s+1}$;
in the proof, we use the well-known Vandermonde determinant.
\begin{Lemma}[Vandermonde determinant]\label{Vandermonde}\hfill\\
  Let $x_1,\,\ldots,\,x_n$ be elements in a field $K$ and consider the matrix $A = (a_{ij})_{i,j = 1}^{n}$
  with $a_{ij} = x_i^{j-1}$. Then, there holds $\det A = \prod_{i > j} (x_i - x_j)$, and $A$ is invertible,
  if and only if all the $x_i$ are different.
\end{Lemma}
\begin{Theorem}[Algebra generated by rep's of $\mathfrak{su}(2)$]\label{GeneratedAlgebra}\hfill\\
  For $s \in \Mg{\frac{1}{2},\,1,\,\frac{3}{2},\,2,\,\dots}$, the (associative) algebra generated by
  any representation of $\mathfrak{su}(2)$ on $\C^d$, $d = 2s+1$, is the full matrix algebra $\cM_d$.
\end{Theorem}
\begin{proof}
  The $d$ matrices $(\mS_z)^n$, with $n \in \{0,\ldots,d-1\}$, including the identity $(\mS_z)^0=\Eins_d$ which we can construct by $[s(s+1) \hbar^2] \Eins_d = \hat{\vec S}^2 = \hat{S}_z^2 + \hat{S}_+\hat{S}_- - \hbar \hat{S}_z$,  are diagonal with eigenvalues $(-s)^n$, $(-s+1)^n$, $\dots$, $s^n$. If we identify
  each diagonal matrix $(\mS_z)^n$ with a column vector in $\C^d$, then we can generate all diagonal
  matrices in $\cM_d$, provided these vectors span $\C^d$, i.e. the $d \times d$ matrix constructed of
  the $d$ column vectors is invertible. Since $\mS_z$ is non-degenerate, this is true by lemma \ref{Vandermonde}.
Further, we can decompose any matrix into diagonal and off-diagonal strips and use $(\mS_\pm)^{n}$
  to shift diagonal matrices to any off-diagonal to construct all of $\cM_d$.
\end{proof}
\par We shall now show by finite construction of the inducing observable subalgebas that any state on a finite-dimensional Hilbert space can have
arbitrary bipartite entanglement properties as long as the notion of locality is chosen appropriately.  Previously, special cases of this theorem proved the existence of tensor product structures for which any pure state is separable~\cite{cunha} and of observables that will detect non-local correlations in any pure state~\cite{delatorre}.  
\begin{Theorem}[Tailored Observables Theorem]\label{tailor}\hfill\\
  Let $\cH = \C^d$ be a Hilbert space  with an orthonormal basis $(\ket{i})_{i = 1}^{d}$, and let
  $d = k \cdot l$ with $k,\,l \in \N$. Then, for every state $\ket{\Psi} = \sum_{i = 1}^{d} c_i \ket{i}$
  and every $\lambda_1,\,\dots,\,\lambda_{\min\Mg{k,\,l}} \in \C$ with
  $\sum_{i = 1}^{d} \betrag{c_i}^2 = \sum_{i = 1}^{\min\Mg{k,\,l}} \betrag{\lambda_i}^2$,
  there exist algebras $\cA$ and $\cB$ satisfying the conditions of theorem \ref{Zanardi}, and a unitary
  operator $U$, such that
  $\ket{\Psi} = U\sum_{i = 1}^{\min\Mg{k,l}} \lambda_i \ket{i}_A \ket{i}_B$ with orthonormal bases
  $(\ket{i}_A)_{i = 1}^{k}$ and $(\ket{i}_B)_{i = 1}^{l}$ of the Hilbert spaces $\cH_A = \C^k$ and $\cH_B = \C^l$, respectively.
\end{Theorem}

\begin{proof}
  Without loss of generality, we assume $k \leq l$; an arbitrarily entangled state on $\C^k \x \C^l$ may then
  be written in its Schmidt form as $\ket{\phi} = \sum_{i = 1}^{k} \lambda_i \ket{i}_A \ket{i}_B$ where
  $(\ket{i}_A)_{i = 1}^{k}$ and $(\ket{i}_B)_{i = 1}^{l}$ are eigenvectors of the $k$- resp. $l$-dimensional representations $\mS^{(A)}_z$ and $\mS^{(B)}_z$ of $\hat{S}_z$.
  Using the Gram-Schmidt procedure, we can obtain orthonormal bases $\Mg{\ket{\Psi_1},\,\dots,\,\ket{\Psi_d}}$
  of $\C^d$ and $\Mg{\ket{\phi_1},\,\dots,\,\ket{\phi_{k \cdot l}}}$ of $\C^k \x \C^l$ with
  $\ket{\Psi_1} = \ket{\Psi}$ and $\ket{\phi_1} = \ket{\phi}$, and we set $U\ket{\phi_i} := \ket{\Psi_i}$ for all $i \in \MgE{d}$. The algebras $\cA$ and $\cB$ are chosen to be generated by the operators
  $U(\mS^{(A)}_j \otimes \Eins)U^\dag$ and $U(\Eins \otimes \mS^{(B)}_j )U^\dag$ for $j \in \Mg{x,y,z}$;
  %by theorem \ref{GeneratedAlgebra} it is clear that they fulfill the conditions from theorem \ref{Zanardi}.
  by theorem \ref{GeneratedAlgebra} they fulfill the conditions from theorem \ref{Zanardi}.
\end{proof}

\emph{Example}: Consider the simplest case: an unstructured Hilbert space $\HS=\mathbb{C}^4$ with basis $\Mg{\ket{0},\ldots,\ket{3}}$, and a pure state $\kt{\Psi}=\kt{0}$~\cite{comment}.  We want to tailor observable subalgebras $\mathcal{A}$ and $\mathcal{B}$ that induce a factorization $\HS_A\otimes\HS_B=\mathbb{C}^2\otimes\mathbb{C}^2$ with respect to which $\kt{\Psi}$  has the same entanglement  as the state $\ket{\phi}=\lambda_1 \ket{0}_{A'}\ket{0}_{B'}+\lambda_2 \ket{1}_{A'}\ket{1}_{B'} $ has with respect to a model Hilbert space $\HS'=\HS_{A'}\otimes\HS_{B'}$ with inducing subalgebras $\mathcal{A}'$ and $\mathcal{B}'$. To do this, 
we first make the identifications $\kt{j}_{A'}\kt{k}_{B'}=\kt{jk}=\kt{2j+k}$, e.g. $\ket{01}=\kt{1}$ and $\ket{11}=\kt{3}$.  Then we
define a unitary operator that maps $\ket{\Psi}$ into the state $\ket{\phi}=U\kt{\Psi}=\lambda_1 \kt{0} + \lambda _2 \kt{3}$. A simple choice is
\begin{equation}
  U = \begin{pmatrix} \lambda_1 & 0 & 0 & \lambda_2 \\ 0 & 1 & 0 & 0 \\
      0 & 0 & 1 & 0 \\ -\lambda_2 & 0 & 0 & \lambda_1 \end{pmatrix}.
\end{equation}
This unitary operator and identification are not unique, and the freedom here could be exploited if there were additional practical constraints on the types of measurements.  To make a connection to the algebra of observables, we define the subalgebra $\mathcal{A}'$ ($\mathcal{B}'$) as the algebra generated by the operators $\mS_j^{A'}=\frac{1}{2}\sigma_j\otimes\Eins$ ($\mS_j^{B'}=\frac{1}{2}\Eins\otimes\sigma_j$) with $j\in\{x,y,z\}$.  The basis vectors $\kt{jk}$ are the joint eigenvectors of $\sigma_z\otimes\Eins$ and $\Eins\otimes\sigma_z$.  Then we can use the operator $U$ to map the subalgebras $\mathcal{A}'$ and $\mathcal{B}'$ back into their tailored representations $\mathcal{A}$ and $\mathcal{B}$ in the original unstructured Hilbert space $\HS$.  For example, the generators  of the subalgebra $\mathcal{A}$ are represented in the unstructured Hilbert space basis  as
\begin{eqnarray}
 U^\dag (\sigma_x\otimes\Eins) U &=& \lambda_1\sigma_x\otimes \Eins +\lambda_2 \sigma_z\otimes\sigma_x\nonumber\\
U^\dag( \sigma_y\otimes\Eins )U &=& \lambda_1\sigma_y\otimes \Eins -\lambda_2 \sigma_z\otimes\sigma_y\nonumber\\
 U^\dag (\sigma_z\otimes\Eins )U &=&  \lambda_1^2\sigma_z\otimes \Eins -\lambda_2^2\Eins\otimes\sigma_z\nonumber\\
&&  - \lambda_1\lambda_2 \sigma_x\otimes\sigma_x +\lambda_1\lambda_2 \sigma_y\otimes\sigma_y
\end{eqnarray}
and the generators of $\mathcal{B}$ can be found by transposing the order of the Pauli matrices in every term.  Non-local operators like $U^\dag(\sigma_z\otimes\sigma_z)U$, required for observing Bell-type inequalities, can also be constructed this way.  In the case of $d=4$, only linear factors of the tensored Pauli matrices appear in the tailored observables; in the case $d>4$, powers of the $\mS_j^{A'}$ and $\mS_j^{B'}$ matrices will appear.
\par Specifying $\lambda_1=\lambda_2=1/\sqrt{2}$, this previous example shows how to construct subalgebras of observables that induce maximal bipartite entanglement in an arbitrary pure state and proves that such an construction can be done in a finite number of steps.  Extending to arbitrary initial states and higher dimensions only requires more computational effort to determine an appropriate $U$, but is in principle no more complicated.  Theorem \ref{tailor} can be extended to multipartite tensor product structures.
\begin{Corollary}[Extended Tailored Observables Theorem]\label{tailor2}
  Given a finite-dimensional Hilbert space $\cH = \C^d$ with an orthonormal basis $(\ket{i})_{i = 1}^{d}$, a state $\ket{\Psi} = \sum_{i = 1}^{d} c_i \ket{i}$, and a factorization \mbox{$d = \prod_{i = 1}^{N} k_i$} with
  $k_i \in \N$, there exist operator algebras $\cA_1,\,\dots,\,\cA_N$ that fulfill the conditions of
  corollary~\ref{ZanardiMulti}, and there exists a unitary mapping $\Fkt{U}{\bigotimes_{i = 1}^{N} \C^{k_i}}{\C^d}$
  satisfying $\cA_i = U(\Eins_{k_1} \x \cdots \x \Eins_{k_{i-1}} \x  \cM_{k_i} \x \Eins_{k_{i+1}}\x\cdots \x \Eins_{k_N})U^\dagger$ for $i \in \MgE{N}$
  such that $\ket{\phi}=U^\dagger \kt{\Psi}$ has an expansion
  \begin{equation}\label{multiexp}
    \ket{\phi} = \sum\nolimits_{i_1 = 1}^{k_1} \dots \sum\nolimits_{i_N = 1}^{k_N}
      \Bigl(c_{i_1 i_2 \ldots i_N} \ket{i}_1 \ket{i}_2 \dots \ket{i}_N\Bigr).
  \end{equation}
\end{Corollary} 
The proof follows from Theorem \ref{tailor} by induction, except one cannot rely on the Schmidt form for effectively classifying the amount of entanglement. Nevertheless, since any pure state entanglement characterization must be able to be expressed in the form (\ref{multiexp}), we have shown that, in principle, for finite-dimensional systems all entanglement properties can be reproduced in any state by choosing the correct observables.  The explicit construction of the unitary operator may require cleverness to accomplish efficiently, but it exists and can be constructed in finite steps.
\par As a final comment, the Tailored Observables Theorem presented here applies only to pure states in finite dimensions. As stated, it cannot extend to mixed states; for example, a totally-mixed state is represented as $\frac{1}{d} \Eins_d$ in any basis, and therefore is totally mixed in any tensor product structure.  Totally mixed states have no coherences that can be shifted to non-local sectors of the tensor product structure and exploited as entanglement.  An open question is how to construct observables that make a partially mixed state as entangled as possible.  Additionally, a full generalization to infinite dimensions and continuous variables is beyond the scope of this Letter but is of practical and intrinsic interest.

\acknowledgments NLH gratefully acknowledges the hospitality of the Institut f\"{u}r Quantenphysik, Universit\"{a}t Ulm, where this work was begun.  KSR is funded by the project Quantum Optics Repeater Platform (QuOReP) by the German Bundesministerium f\"{u}r Bildung und Forschung (BMBF).

\end{document}